\documentclass[12pt,a4paper]{article}
\usepackage{jhep-mod}
\usepackage{datetime}

\newif\iffinal

\usepackage[numbers,sort&compress]{natbib}
\usepackage{doi} 

\usepackage{physics}
\usepackage{amsthm}

	\newtheorem{lemma}{Lemma}
	\AfterEndEnvironment{definition}{\noindent\ignorespaces}
	\AfterEndEnvironment{theorem}{\noindent\ignorespaces}
	\AfterEndEnvironment{lemma}{\noindent\ignorespaces}
\usepackage{dsfont}
\usepackage{mathtools}

\newenvironment{subalign}[1][]{\subequations{#1}\align}{\endalign\endsubequations}

\newcommand{\toself}[1]{\iffinal\else\textcolor{blue}{\{\textsc{to self:} #1\}}\fi}

\makeatletter
\newcommand\newshortcut[2]{\@namedef{\detokenize{#1}}{#2}}
\newcommand\x[1]{\@nameuse{\detokenize{#1}}}
\makeatother

\renewcommand{\vb}[1]{\boldsymbol{#1}}
\newcommand{\ve}[1]{\vb{e}_{#1}}
\newcommand{\vg}[1]{\vb{\gamma}_{#1}}
\newcommand{\vs}[1]{\vb{\sigma}_{#1}}

\renewcommand{\ip}[2]{\left\langle#1,#2\right\rangle}

\newcommand{\GA}[1][]{\mathcal{G}_{#1}}

\usepackage{dsfont}

\newcommand{\rotor}[1]{\mathcal{#1}}
\newcommand{\linmap}[1]{\mathrm{#1}}

\usepackage{mathbbol}

\newcommand{\vol}{\mathbb{i}}

\newcommand{\spanof}[2][]{\operatorname{span}_{#1}\qty{#2}}
\newcommand{\grade}[2][]{\left\langle#2\right\rangle_{#1}}
\newcommand{\rev}[1]{\widetilde{#1}}

\newcommand{\srev}[1]{\qty{\!\!\qty{#1}\!\!}}
\newcommand{\arev}[1]{\qty[\!\qty[#1]\!]}

\newcommand{\RR}{\mathds{R}}
\newcommand{\CC}{\mathds{C}}
\newcommand{\HH}{\mathds{H}}
\newcommand{\ZZ}{\mathds{Z}}
\newcommand{\NN}{\mathds{N}}

\newcommand{\RE}[1]{\mathfrak{R}\qty(#1)}
\newcommand{\IM}[1]{\mathfrak{I}\qty(#1)}

\DeclareMathOperator{\SO}{SO}
\DeclareMathOperator{\Spin}{Spin}

\newcommand{\so}{\mathfrak{so}}

\DeclareMathOperator{\diag}{diag}

\newcommand{\Co}{\mathrm{C}_}
\newcommand{\Si}{\mathrm{S}_}
\newcommand{\Ta}{\mathrm{T}_}

\newcommand{\bch}[2]{#1 \circledcirc #2}

\newcommand{\orcidicon}{
	\includegraphics[width=1em]{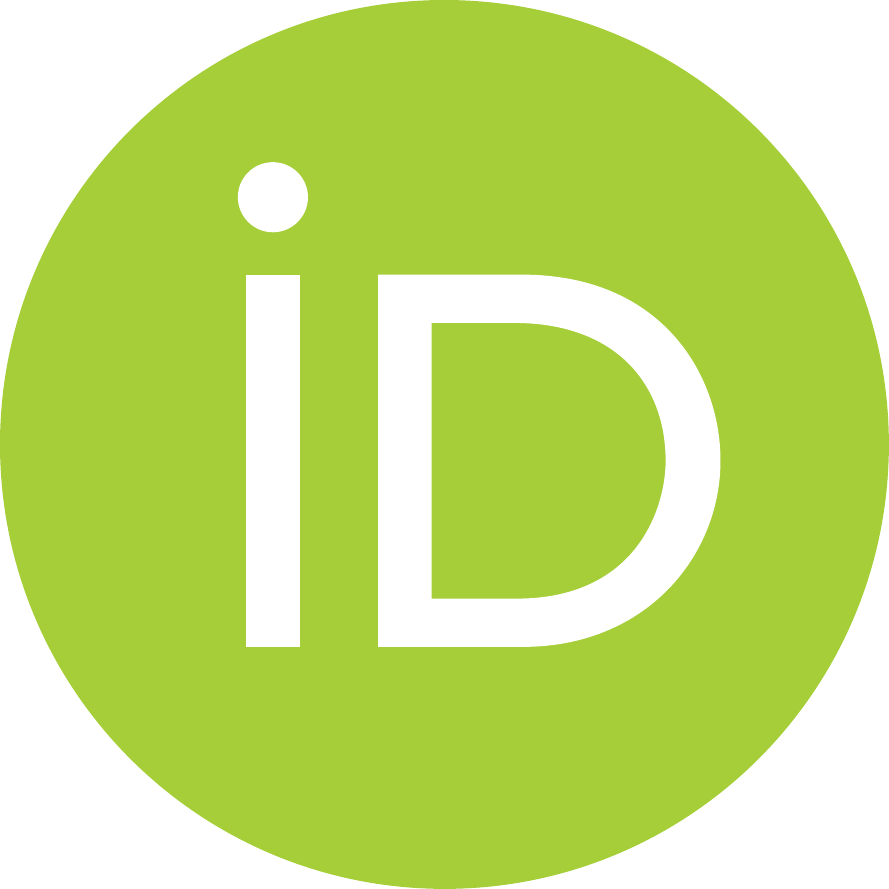}
}
\newcommand\orcidJoseph{{\href{https://orcid.org/0000-0002-4693-6539}{\orcidicon}}}
\newcommand\orcidMatt{{\href{https://orcid.org/0000-0003-1088-6485}{\orcidicon}}}

\finaltrue 

\iffinal
\else
\usepackage{libertine}
\usepackage{fontspec}
\fi

\newshortcut{BCH full}{Baker--Campbell--Hausdorff--Dynkin}
\newshortcut{BCH}{BCHD}

\title{Explicit \x{BCH full} formula for Spacetime via Geometric Algebra}

\author{Joseph Wilson \orcidJoseph}
\author{{\sf and} Matt Visser \orcidMatt}

\affiliation{
	School of Mathematics and Statistics, Victoria University of Wellington,
\\	\null\qquad PO Box 600, Wellington 6140, New Zealand
\\[2ex]	\toself{Compiled on \today\ at \currenttime.}
}

\emailAdd{joseph.wilson@sms.vuw.ac.nz}
\emailAdd{matt.visser@sms.vuw.ac.nz}

\abstract{
	We present a compact \x{BCH full} formula for the composition of Lorentz transformations $e^{\sigma{}_i}$ in the spin representation (\emph{a.k.a.} Lorentz rotors) in terms of their generators $\sigma{}_i$:
	\begin{align*}
		\ln(e^{\sigma{}_1}e^{\sigma{}_2}) =
		\tanh^{-1}\qty(\frac{
			\tanh \sigma{}_1 + \tanh \sigma{}_2 + \frac12[\tanh \sigma{}_1, \tanh \sigma{}_2]
		}{
			1 + \frac12\qty{\tanh \sigma{}_1, \tanh \sigma{}_2}
		})
	\end{align*}
	This formula is general to geometric algebras (\emph{a.k.a.}\ real Clifford algebras) of dimension $\leq{} 4$, naturally generalising Rodrigues' formula for rotations in $\RR^3$.
	In particular, it applies to Lorentz rotors within the framework of Hestenes' spacetime algebra, and provides an efficient method for composing Lorentz generators.
	Computer implementations are possible with a complex $2\times{}2$ matrix representation realised by the Pauli spin matrices.
	The formula is applied to the composition of relativistic $3$-velo\-cities yielding simple expressions for the resulting boost and the concomitant Wigner angle.
}

\begin{document}

\maketitle

\clearpage

\section{Introduction}
\label{sec:intro}

In studying proper Lorentz transformations, it is often easier to represent them in terms of their generators $\sigma{}_i$ belonging to the Lorentzian Lie algebra $\mathfrak{so}(1,3)$.
A fundamental question is how Lorentz transformations compose in terms of these generators: ``given $\sigma{}_1$ and $\sigma{}_2$, what is $\sigma{}_3$ such that $e^{\sigma{}_1}e^{\sigma{}_2} = e^{\sigma{}_3}$?''
This is useful theoretically and in practical applications, where representing transformations in terms of their generators is cheaper.
One may use the \x{BCH full} (\x{BCH}) formula $\bch{\sigma{}_1}{\sigma{}_2} :={} \ln(e^{\sigma{}_1}e^{\sigma{}_2})$ which is well studied in general Lie theory \cite{achilles2012bch-early}.
However, the general \x{BCH} formula
\begin{align}
	\bch{a}{b} = a + b + \frac12[a, b] + \frac{1}{12}[a, [a, b]] + \frac{1}{12}[[a, b], b] + \cdots
	\label{eqn:general-BCH-formula}
\end{align}
involves an infinite series of nested commutators and so is not immediately of practical use for the study of Lorentz transformations.
Some closed-form expressions for \eqref{eqn:general-BCH-formula} under the $2$-form representation of $\so(1,3)$ have been found \cite{coll2002sr-generator-composition,coll1990sr-generator-exponentiation}, but the expressions are complicated and do not clearly reduce to well-known formulae in, for example, the special cases of pure rotations or pure boosts.

We present a relatively simple closed-form \x{BCH} formula for orthogonal transformations in any space of dimension $\leq{}4$ within the framework of geometric algebra.
In the case of Lorentzian spacetime, we point out how the $2\times{}2$ complex linear representation enabled by the Pauli spin matrices provides an efficient method to numerically or symbolically compose Lorentz transformations in terms of their generators.
The formula is of sufficient simplicity to be of pedagogical interest, easily yielding standard results for the composition of relativistic $3$-velocities and the associated Wigner angle \cite{berry2020quat-sr,visser2011sr-velocity-composition,wigner1939unitary}.

\subsection{Geometric algebra: Historical context and motivation}

The basic ingredient of geometry in the context of classical physics and special relativity is a real vector space $V = \RR^n$ representing a (local) frame in physical space.
This space is equipped with a (possibly indefinite) vector inner product, or \emph{metric}\footnote{
	Mathematically, a \emph{symmetric bilinear form}.
}
\begin{align}
	\eta{}(\vb{u}, \vb{v}) \equiv{} \ip{\vb{u}}{\vb{v}} \equiv{} \eta{}_{ab}\,u^a\,v^b \in \RR	
\end{align}
which gives rise to notions like length and angle.
Denote by $\RR^{p,q}$ the real $(p + q)$-dimensional vector space of signature\footnote{
	That is, with metric $\eta{} \cong \diag(\overbrace{+1,\cdots+1}^p,\overbrace{-1,\cdots,-1}^q)$ in the standard basis.
	We shall not consider degenerate signatures, although these find use in computer graphics as \emph{projective geometric algebras} \cite{hestenes1991pga,vince2008ga-graphics} and in areas of general relativity involving light-like hypersurfaces \cite{poisson2002gr-null-shells,israel1966gr-hypersurfaces}.
} $(p, q)$.
The isometries of $\RR^{p,q}$ taken together form the orthogonal groups $\SO(p,q)$ and describe transformations between \emph{inertial observers} in that space.
These groups (along with their Lie algebras $\so(p,q)$, associated spin groups $\Spin(p,q)$ and linear representations) are central objects of study in relativistic and quantum physics.

Historically, there was tension regarding the best algebraic framework in which to express these geometric objects.
The vector algebra ``war'' of 1890--1945 saw Hamilton's quaternions $\HH$, once hailed as the optimal language for describing $\RR^3$ rotations, lose popularity in favour of Gibbs' 3-vectors.
For their elegant handling of $\RR^3$ rotations, many authors have tried coercing quaternions into $\RR^{1,3}$ for application to special relativity \cite{silberstein1912quat-sr,deleo1996quat-sr,dirac1944quat-sr}.
This has been done in various ways, usually by complexifying $\HH$ into an eight-dimensional algebra $\CC\otimes\HH$ and then restricting the number of degrees of freedom as seen fit \cite{berry2020quat-sr,berry2021quat-sr}.
However, it is fair to say that quaternionic formulations of special relativity never gained notable traction.
Today, the  physics community is most familiar with tensor calculus, differential forms and the Dirac $\gamma{}$-matrix formalism, and has relatively little to do with quaternions or quaternion-like algebras \cite{chappell2016quat-history,altmann1989quat-history}.

Arguably, this outcome of history is unfortunate, because both approaches --- tensorial and quat\-ern\-ionic --- possess some advantages over the other.
While quaternionic algebras describe rotations with maximal efficiency, they are clouded by a history of inconsistent interpretations and odd notational choices.
For instance, a notable defect in Hamilton's original presentation of his quaternions was that, because they naturally represent rotation operators, pure quaternions are actually \emph{bivectors} (\emph{i.e.}, ``axial'' vectors or ``pseudovectors'') rather than true (``polar'') vectors.
This leads to a mysterious negative Pythagorean norm and makes the interpretation of pure quaternions as $\RR^3$ vectors misleading \cite{chappell2016quat-history}.
On the other hand, while the usual tensor formalisms represent vectors, bivectors, and general tensors with great ease, their explicit algebraic description of rotations is prohibitively cumbersome.

It appears that geometric algebra, being a unified language for both vectorial objects and rotation operators, provides a comprehensive framework for geometry in physics \cite{chappell2016quat-history,doran2003ga,doran1994ga,lasenby2016ga-unified-language,hestenes1986ga-unified-language,gull1993sta}.
Geometric algebra intrinsically describes rotations as \emph{rotors} which exist in the \emph{spin representation}, leading to a great simplification of formulae involving rotations, and in particular, to a useful geometric \x{BCH} formula.

\subsection{Notations and terminologies in geometric algebra}

A self-contained primer on geometric algebra is given in appendix~\ref{apx:geoalg}.
See also \cite{hestenes2003sta,doran2003ga,lasenby2016ga-unified-language} for introductions aimed at a physics audience.
We mostly adopt the notations of \cite{doran2003ga}, denoting by $\GA(V, \eta{}) \equiv{} \GA(p, q)$ the geometric algebra (\emph{a.k.a.}~the real Clifford algebra) over a real vector space $V$ of dimension $p + q$ equipped with a metric $\eta{}$ of signature $(p, q)$.
Denote the \emph{subspace of grade $k$} by $\GA[k](p, q)$, and the \emph{even sub-algebra} by $\GA[+](p, q) \equiv{} \bigoplus_{k}\GA[2k](p, q)$.

Generic elements of $\GA(p,q)$ are called \emph{multivectors}, and homogeneous multivectors of a fixed grade $k$ are called \emph{$k$-vectors}.
Write the grade $k$ projection of $a$ as $\grade[k]{a}$.
Grade one vectors are denoted in boldface, `$\vb{v}{}$'; and even multivectors (usually rotors) are in script, `$\rotor R$'.
Finally, a multivector $A$ which is a sum of $k$-vectors for $k \in K$ is called a \emph{$K$-multivector} (e.g., the sum of a scalar and a $k$-vector is a $\qty{0, k}$-multivector).

We denote by $\rev{a}$ the \emph{reverse} of $a$ (where the order of the geometric product is reversed, as in $\rev{ab} = \rev{b}\rev{a}$), and the \emph{volume element} in $n$ dimensions by $\vol$.
In $\GA(p, q)$ generally, $\vol^2 = \pm 1$, but in Euclidean 3-space $\GA(3)$ and in spacetime $\GA(1,3)$, the pseudoscalar satisfies $\vol^2 = -1$ as suggested by its symbol.

\subsection{Bivectors, rotors and Lorentz transformations}

It is worth noting the relationship between orthogonal transformations of a vector space and their analogous description as \emph{rotors}, which belong to the double-cover of the orthogonal group (see \cite[\textsection{}\,11.3]{doran2003ga} and \cite{lasenby2016ga-unified-language,hestenes1986ga-unified-language}.)
The advantage of this additional formalism is hopefully clear: It leads to an elegant and unified\footnote{
	Rotors are `elegant' because: they are always of the form $\pm e^\sigma{}$ for a bivector $\sigma{}$ carrying clear geometric meaning \cite[\textsection{}\,11.3]{doran2003ga}; their description of rotations is free of gimbal lock, and; they can always be interpolated without ambiguity \cite{lasenby2011ga-practical}.
	The formalism is `unified' because rotors: eliminate the need for special treatment of spinors \cite{hestenes1986ga-unified-language}; are general to a space of any dimension or signature, and; act on objects of all grades via the single transformation law \eqref{eqn:double-sided-transformation-law}.
	Of mathematical significance: any finite Lie group is realised as a rotor group, and every Lie algebra as a set of bivectors \cite{lie-groups-as-spin-groups}.
} treatment of generalised rotations, including Lorentz transformations.

An orthogonal transformation in $n$ dimensions may be achieved by the composition of at most $n$ reflections.\footnote{This is the Cartan--Dieudonn\'{e}{} theorem \cite{cartan-dieudonne-theorem}.}
In geometric algebra, a multivector $A \in{} \GA(p,q)$ is reflected across the vector $\vb{v}{}$ by the map
\begin{align}
	A \mapsto{} -\vb{v}{}A\vb{v}{}
	\label{eqn:reflection}
\end{align}
where $\vb{v}{}^2 = \pm{}1$.
A parity-preserving orthogonal transformation $\linmap R \in{} \SO(p, q)$ is therefore achieved by an even number of reflections \eqref{eqn:reflection}, resulting in a \emph{double-sided transformation}
\begin{align}
	\linmap R : A \mapsto \rotor R A \rev{\rotor R}
	\label{eqn:double-sided-transformation-law}
\end{align}
where $\rotor R = \vb{v}{}_1\vb{v}{}_2\cdots \vb{v}{}_{2k} \in \GA[+](p,q)$ is a multivector product of vectors satisfying $\rotor R\rev{\rotor R} = \pm{}1$.
For concreteness, the matrix components of $\linmap R$ with respect to a basis $\qty{\ve a}$ may be obtained from the rotor by $\linmap R^a{}_b = \eta{}^{ac}\ip{\ve c}{\rotor R \ve b \rev{\rotor R}}$ where $\eta{}_{ab} = \ip{\ve a}{\ve b}$ raises and lowers indices.

Under the geometric product, these even multivector products form the \emph{spin group}
\begin{align}
	\Spin(p, q) :={} \qty{\rotor R \in \GA[+](p, q) \mid \rotor R\rev{\rotor R} = \pm{}1} \twoheadrightarrow \SO(p, q)
.\end{align}
To any transformation $\linmap R \in \SO(p, q)$, there are exactly two multivectors of $\Spin(p, q)$ which generate $\linmap R$ under \eqref{eqn:double-sided-transformation-law}, namely $+\rotor R$ and $-\rotor R$, making the spin group a double cover of $\SO(p,q)$, signified by $\Spin(p, q) \twoheadrightarrow \SO(p, q)$.
The even-dimensional linear representations of $\Spin(p, q)$ map to linear representations of $\SO(p, q)$, while odd-dimensional representations map to \emph{projective} representations of $\SO(p,q)$ (or as commonly known in physics, the \emph{spinorial} representations).
Thus, all the tensorial and spinorial transformation laws utilised in physics are realised by the representation theory of multivectors in the spin group.
When dealing with non-spinorial vector representations of $\SO(p, q)$, the overall sign of a rotor is redundant because it does not affect the associated orthogonal transformation.

The further restriction that $\rotor R\rev{\rotor R} = +1$ defines the identity-connected component\footnote{
	Except for the degenerate $(1 + 1)$-dimensional case, $p = q = 1$ \cite{doran2003ga}.
} of the spin group, called the \emph{rotor group}
\begin{align}
	\Spin^+(p, q) :={} \qty{\rotor R \in \GA[+](p, q) \mid \rotor R\rev{\rotor R} = +1}
	\twoheadrightarrow
	\SO^+(p, q)
,\end{align}
whose elements $\rotor R$ are called \emph{rotors}.
In Euclidean spaces, $\SO(n)$ is connected, and so there is no distinction between the spin and rotor groups --- but in mixed signature spaces, $\Spin(p,q) = \Spin^+(p, q) \times{} \ZZ_2$.
All rotors of $\Spin^+(p, q)$ are of the form $\rotor R = \pm e^\sigma{}$ for some bivector generator $\sigma{} \in{} \GA[2](p, q)$ \cite[\textsection{}\,11.3.3]{doran2003ga}.
Indeed, the subspace of bivectors $\GA[2](p, q)$ forms a Lie algebra under the commutator, with the exponential map $\sigma{} \mapsto{} e^\sigma{}$ sending bivectors to rotors in $\Spin^+(p, q)$.
This Lie algebra is isomorphic to the Lie algebra $\so(p, q)$ of the (special) orthogonal group.
Thus, the Lie algebraic description of generalised rotations by their generators is embedded in the bivector subspace $\GA[2](p, q)$ of the unified geometric algebra.

Conveniently, for (anti-)Euclidean spaces (where either $p$ or $q$ is zero) and the special case of Minkowski spacetime $\qty{p, q} = \qty{1, 3}$, every rotor $\rotor R \in{} \Spin^+(p, q)$ is of the form $\rotor R = e^\sigma{}$ (see \textsection{}\,\ref{sec:invariant-bivector-decomposition}).
We shall keep our main result general to $\leq{} 4$ dimensions, but with emphasis on the case of Minkowski spacetime $\RR^{1,3}$.
In this context, proper orthochronous Lorentz transformations $\linmap \Lambda{} \in \SO^+(1,3)$ are represented by rotors $e^\sigma{} \in \Spin^+(1,3)$, which are in turn generated by spacetime bivectors $\sigma{} \in \GA[2](1,3)$.

\section{A Geometric \x{BCH full} Formula}

Suppose $\sigma{} \in \GA[2](p, q)$ is a bivector in a geometric algebra of dimension $p + q \leq{} 4$.
By their definitions as formal power series, we have
\begin{math}
	e^{\sigma{}} = \cosh \sigma{} + \sinh \sigma{}
,\end{math}
where `$\cosh$' involves even powers of $\sigma{}$ and `$\sinh$' odd powers.
For convenience, define the linear projections onto \emph{self-reverse} and \emph{anti-self-reverse} parts respectively as
\begin{align}
	\srev{A} &= \frac12\qty(A + \rev{A})
&	&\text{and}
&	\arev{A} &= \frac12\qty(A - \rev{A})
	\label{eqn:rev-notation}
.\end{align}
Since any bivector obeys $\rev{\sigma{}} = -\sigma{}$, it follows that $\rev{e^\sigma{}} = e^{-\sigma{}} = \cosh \sigma{} - \sinh \sigma{}$.
Using the notation \eqref{eqn:rev-notation}, the self-reverse and anti-self-reverse projections of $e^\sigma{}$ are $\srev{e^{\sigma{}}} = \cosh \sigma{}$ and $\arev{e^{\sigma{}}} = \sinh \sigma{}$, respectively.
Furthermore, these two projections commute, and so
\begin{align}
	{\arev{e^{\sigma{}}}}{\srev{e^{\sigma{}}}}^{-1} =
	{\srev{e^{\sigma{}}}}^{-1}{\arev{e^{\sigma{}}}} = \tanh \sigma{}
\end{align}
which leads to an expression for the logarithm of any rotor $\rotor R = \pm e^\sigma{}$.
\begin{align}
	\sigma{} = \ln(\rotor R) = \tanh^{-1}\qty(\frac{\arev{\rotor R}}{\srev{\rotor R}})
	\label{eqn:log-rotor}
\end{align}
Note that the overall sign of the rotor is not recovered, and $\ln(+\rotor R) = \ln(-\rotor R)$ according to \eqref{eqn:log-rotor}.
Since $\pm \rotor R$ both represent the same transformation $\linmap R \in \SO^+(p, q)$, this does not affect vector representations, but becomes important when considering spinors.
The exact sign can be recovered by considering the relative signs of $\arev{\rotor R}$ and $\srev{\rotor R}$, as in \cite[\textsection{}\,5.3]{lasenby2011ga-practical}.

From this we may derive a \x{BCH} formula by substituting $\rotor R = e^{\sigma{}_1}e^{\sigma{}_2}$ for any two bivectors $\sigma{}_i \in \GA[2](p, q)$.
Using the shorthand $\Co{i} :={} \cosh \sigma{}_i$ and $\Si{i} :={} \sinh \sigma{}_i$, the composite rotor is
\begin{align}
	\rotor R = e^{\sigma{}_1}e^{\sigma{}_2}
	= (\Co1 + \Si1)(\Co2 + \Si2)
	= \Co1\Co2 + \Si1\Co2 + \Co1\Si2 + \Si1\Si2
.\end{align}
For $p + q < 4$, any even function of a bivector (such as $\Co{i}$) is a scalar, while for $p + q = 4$ the result is a $\qty{0,4}$-multivector $\alpha{} + \beta{}\vol$.
In either case, the $\Co{i}$ commute with even multivectors; $[\Co{i}, \Co{j}] = [\Co{i}, \Si{j}] = 0$.
Therefore, the self-reverse and anti-self-reverse parts are
\begin{align}
	\srev{\rotor R} &= \Co1\Co2 + \frac12\qty{\Si1, \Si2}
&	&\text{and}
&	\arev{\rotor R} &= \Si1\Co2 + \Co1\Si2 + \frac12\qty[\Si1, \Si2]
	\label{eqn:arev-and-srev-parts}
.\end{align}
Hence, from \eqref{eqn:log-rotor} we obtain an explicit \x{BCH} formula
\begin{align}
	\bch{\sigma{}_1}{\sigma{}_2}
	= \tanh^{-1}\qty(\frac{
		\Ta1 + \Ta2 + \frac12\qty[\Ta1, \Ta2]
	}{
		1 + \frac12\qty{\Ta1, \Ta2}
	})
	\label{eqn:geometric-BCH}
\end{align}
where we abbreviate $\Ta{i} \coloneqq \tanh \sigma{}_i$.

We may wish to express \eqref{eqn:geometric-BCH} in terms of geometrically significant products instead of \mbox{(anti-)}commutators.
The geometric product of two bivectors $a$ and $b$ is generally a $\qty{0,2,4}$-multivector
\begin{align}
	ab = \grade[0]{ab} + \grade[2]{ab} + \grade[4]{ab}
.\end{align}
Employing the notation of Hestenes \cite{doran2003ga}, this may be written as
\begin{align}
	ab = a\cdot{}b + a\times{}b + a\wedge b
	\label{eqn:bivector-products}
,\end{align}
where here $a\times{}b = \grade[2]{ab} = \frac12(ab - ba)$ is the bivector \emph{commutator product}, and the scalar inner product $a\cdot{}b = \grade[0]{ab}$ is extended to bivectors.
We may then write a \x{BCH} formula in which the grade of each term is explicit:
\begin{align}
	\bch{\sigma{}_1}{\sigma{}_2} = \tanh^{-1}\qty(\frac{
		\Ta1 + \Ta2 + \Ta1 \times{} \Ta2
	}{
		1 + \Ta1 \cdot{} \Ta2 + \Ta1 \wedge{} \Ta2 
	})
	\label{eqn:geometric-BCH-products}
\end{align}
The numerator is a bivector, while the denominator contains scalar ($\Ta1\cdot{}\Ta2$) and $4$-vector ($\Ta1\wedge\Ta2$) terms.


\subsection{Specialisation in low dimensions}

It is illustrative to see how the \x{BCH} formula \eqref{eqn:geometric-BCH} reduces in the two- and three-dimensional special cases.

\subsubsection{2D: The Euclidean and hyperbolic plane}

In two dimensions, all bivectors are scalar multiples of $\vol = \ve1\ve2$, and we recover the trivial case $e^ae^b = e^{a+b}$. 
Specifically, in the Euclidean $\GA(2)$ plane (or anti-Euclidean $\GA(0,2)$ plane) we have $\vol^2 = -1$, and equation \eqref{eqn:geometric-BCH} simplifies by way of the tangent angle addition identity
\begin{align}
	\tan^{-1}\qty(\frac{\tan \theta_1 + \tan \theta_1}{1 - \tan \theta_1 \tan \theta_2}) = \theta_1 + \theta_2
.\end{align}
This identity encodes how angles add when given as the gradients of lines; $m = \tan \theta$.

Similarly, in the hyperbolic plane $\GA(1,1)$ with basis $\qty{\ve+, \ve-}, \ve\pm^2 = \pm1$, the pseudoscalar $\vol = \ve+\ve-$ generates \emph{hyperbolic} rotations $e^{\xi{}\vol} = \cosh \xi{} + \vol\sinh \xi{}$ owing to the fact that $\vol^2 = -\ve+^2\ve-^2 = +1$.
Then, formula \eqref{eqn:geometric-BCH} simplifies by the hyperbolic angle addition identity
\begin{align}
	\tanh^{-1}\qty(\frac{\tanh \xi{}_1 + \tanh \xi{}_1}{1 + \tanh \xi{}_1 \tanh \xi{}_2}) = \xi{}_1 + \xi{}_2
\end{align}
which encodes how collinear rapidities add when given as relativistic velocities; $\beta{} = \tanh \xi{}$.

\subsubsection{3D: Rodrigues' rotation formula}

Less trivially, a rotation in $\RR^3$ by $\theta$ may be represented by its \emph{Rodrigues vector} $\vb r = \vb{\hat{r}}\tan\frac\theta2$ pointing along the axis of rotation.
The composition of two rotations is then succinctly encoded in \emph{Rodrigues' rotation formula}
\begin{align}
	\vb r_{12} = \frac{\vb r_1 + \vb r_2 - \vb r_1 \times{} \vb r_2}{1 - \vb r_1 \cdot{} \vb r_2}
	\label{eqn:rodrigues-formula}
\end{align}
involving the standard vector dot and cross products.

We can easily derive \eqref{eqn:rodrigues-formula} as a special case of \eqref{eqn:geometric-BCH-products} as follows:
Let $\sigma{}_1, \sigma{}_2 \in \GA[2](3)$ be two bivectors defining the rotors $e^{\sigma{}_1}$ and $e^{\sigma{}_2}$ in three dimensions.
In $\GA(3)$, the only $4$-vector is trivial, so $\sigma{}_1 \wedge{} \sigma{}_2 = 0$ and for the composite rotor $e^{\sigma{}_3} \coloneqq e^{\sigma{}_1}e^{\sigma{}_2}$ we have
\begin{align}
	\sigma{}_3 = \bch{\sigma{}_1}{\sigma{}_2} = \tanh^{-1}\qty(\frac{
	\tanh \sigma{}_1 + \tanh \sigma{}_2 + \tanh \sigma{}_1 \times \tanh \sigma{}_2
	}{
	1 + \tanh \sigma{}_1 \cdot{} \tanh \sigma{}_2
	})
\end{align}
where $a \times{} b$ is the commutator product of bivectors as in \eqref{eqn:bivector-products}, not the vector cross product.
Observe that Euclidean bivectors $\sigma{}_i \in{} \GA[2](3)$ have negative square (e.g., $(\ve1\ve2)^2 = -\ve1^2\ve2^2 = -1$) and that they are related to their dual normal vectors $\vb{u}_i$ by $\sigma{}_i = \vb{u}_i\vol$.
Therefore, by rewriting
\begin{math}
	\tanh \sigma{}_i
	= \tanh (\vb{u}_i\vol)
	= (\tan \vb{u}_i)\vol
,\end{math}
we obtain the formula in terms of plain vectors and the vector cross product.
\begin{align}
	\vb{u}_{12} = (\bch{\vb{u}_1\vol}{\vb{u}_2\vol})\vol^{-1}
	= \tan^{-1}\qty(\frac{
	\tan \vb{u}_1 + \tan \vb{u}_2 - \tan \vb{u}_1 \times \tan \vb{u}_2
	}{
	1 - \tan \vb{u}_1 \cdot{} \tan \vb{u}_2
	})
\end{align}
Indeed, a bivector $\sigma{}_i = \vb{u}_i\vol$ generates an $\RR^3$ rotation through an angle $\theta = 2\|\vb{u}_i\|$ via the double-sided transformation law
\begin{math}
	a \mapsto e^{\vb{u}\vol}ae^{-\vb{u}\vol}
.\end{math}
Hence, $\tan \vb{u}_i = \vb{\hat{v}}_i\tan\frac\theta2 \equiv{} \vb{r}_i$ are exactly the half-angle Rodrigues vectors and we recover \eqref{eqn:rodrigues-formula}.
The necessity of the half-angle in the Rodrigues vectors reflects the fact that they actually generate \emph{rotors}, not rotations directly, and so belong in the underlying spin representation of $\SO(3)$ --- a fact made clearer in the context of geometric algebra.

\subsection{In higher dimensions}

In fewer than four dimensions, the $4$-vector $\Ta1\wedge\Ta2 = 0$ appearing in the geometric \x{BCH} formula \eqref{eqn:geometric-BCH-products} is trivial, and so \eqref{eqn:geometric-BCH} involves only bivector addition and scalar multiplication.
In four dimensions, there is one linearly independent $4$-vector --- the pseudoscalar --- which necessarily commutes with all even multivectors.
However, in more than four dimensions, $4$-vectors do \emph{not} necessarily commute with bivectors, and the assumptions underlying \eqref{eqn:arev-and-srev-parts} and hence the main result \eqref{eqn:geometric-BCH} fail.

On the face of it, the \x{BCH} formula \eqref{eqn:geometric-BCH} in the four-dimensional case appears deceptively simple --- it hides complexity in the calculation of the trigonometric functions
\begin{align}
	\tanh \sigma{}_i &= \sigma{} - \frac13\sigma{}^3 + \frac{2}{15}\sigma{}^5 + \cdots
&	&\text{and}
&	\tanh^{-1} \sigma{}_i &= \sigma{} +\frac13\sigma{}^3 + \frac15\sigma{}^5 + \cdots
	\label{eqn:trig-power-series}
\end{align}
of arbitrary bivectors.
In fewer dimensions, $\sigma{}^2$ is a scalar, and so these power series are as easy to compute as their real equivalents (if $N_\sigma{} \in \RR$ satisfies $\sigma{}^2 = N_\sigma{}^2$ then $\tanh \sigma{} = (\tanh N_\sigma{})N_\sigma{}^{-1}\sigma{}$).
But in four dimensions, $\sigma{}^2$ is in general a $\qty{0,4}$-multivector (by lemma~\ref{lem:grades-of-square} of appendix~\ref{apx:geoalg}) and the power series \eqref{eqn:trig-power-series} are more complicated.
However, if $\sigma{}^2 \ne 0$ has a square root $N_\sigma{} = \alpha{} + \beta{}\vol$ in the scalar-pseudoscalar plane, then one has $\sigma{} = N_\sigma{}\hat{\sigma{}} = \hat{\sigma{}}N_\sigma{}$ where $\hat{\sigma{}} \coloneqq \sigma{}/N_\sigma{}$ is `normalized' so that $\hat{\sigma{}}^2 = 1$.
With a bivector $\sigma{} = N_\sigma{}\hat{\sigma{}}$ expressed in this form, the valuation of a formal power series $f(z) = \sum_{n=1}^\infty f_n z^n$ simplifies to
\begin{subalign}[\label{eqn:normalized-power-series}]
	\text{($f$ even)}&
&	f(\sigma{}) &= \sum_{n = 1}^\infty f_{2n} \sigma{}^{2n}
	= \sum_{n = 1}^\infty f_{2n} N_\sigma{}^{2n}
	= f(N_\sigma{})
,\\	\text{($f$ odd)}&
&	f(\sigma{}) &= \sum_{n = 1}^\infty f_{2n + 1} \sigma{}^{2n + 1}
	= \sum_{n = 1}^\infty f_{2n} N_\sigma{}^{2n + 1} \hat{\sigma{}}
	= f(N_\sigma{})\hat{\sigma{}}
.\end{subalign}
This is especially useful in the case of Minkowski spacetime $\GA(1,3)$ because the scalar-pseudoscalar plane $\GA[0,4](1,3) \cong \CC$ is isomorphic to the complex plane, and therefore a square root of $\sigma{}^2$ always exists.
Furthermore, complex square roots and complex trigonometric functions are easily computable.
From now on, we focus on the special case of Minkowski spacetime, considering its practical and theoretical application.

\section{The Algebra of Spacetime}

\emph{Spacetime algebra} (STA) is the name given to the geometric algebra of Minkowski space, $\GA(\RR^4, \eta{}) \equiv{} \GA(1,3)$, where $\eta{} = \pm\diag(-1,+1,+1,+1)$.
Introductory material on the STA can be found in \cite{hestenes2003sta,gull1993sta,dressel2015sta}.

Denote the standard vector basis by $\qty{\vb \gamma{}_\mu{}}$, where Greek indices run over $\qty{0,1,2,3}$.
(This is a deliberate allusion to the Dirac $\gamma{}$-matrices, whose algebra is isomorphic to the STA --- however, the $\vg \mu{}$ of STA are real, genuine spacetime vectors.)
A basis for the entire STA is then
\begin{align*}
	\overset{\text{scalars}}{\qty{\vb 1}}
\cup{}	\overset{\text{vectors}}{\qty{\vg 0, \vg i}}
\cup{}	\overset{\text{bivectors}}{\qty{\vg 0\vg i, \vg j\vg k}}
\cup{}	\overset{\substack{\text{trivectors}}}{\qty{\vg 0\vg j\vg k, \vg 1\vg 2\vg 3}}
\cup{}	\overset{\substack{\text{pseudoscalar}}}{\qty{\vol :={} \vg0\vg1\vg2\vg3}}
\end{align*}
where Latin indices range spacelike components, $\qty{1,2,3}$.
Multivectors constructed in a basis-invariant manner are manifestly Lorentz-invariant quantities.

The right-handed unit pseudoscalar $\vol$ represents an oriented volume element and satisfies $\vol^2 = -1$.
This is one way in which complex structure arises within the real STA.
The scalar--pseudoscalar plane $\GA[0,4](1,3) = \spanof[\RR]{1, \vol}$ is algebraically isomorphic to the complex plane $\CC$, and so for the sake of computation, $\qty{0, 4}$-multivectors may be simply regarded as complex numbers.
In particular, we define the principal root $\sqrt{a}$ of a $\qty{0,4}$-multivector $a \in \GA[0,4](1,3)$ in the same way as it is defined in $\CC$ with a branch cut at $\theta = \pi{}$.
It is worth emphasising that there are many square roots of $-1$ in the spacetime algebra, with distinct geometrical meanings.
(For instance, a spacelike bivector $(\vg{i}\vg{j})^2 = -1$ represents a directed spacelike plane.)
We have chosen to define ``$\sqrt{\phantom{a}}$'' in such a way that $\sqrt{-1} = \vol$ is singled out as the principal root, as this proves to be useful.\footnote{
	Especially in electromagnetic theory, the imaginary unit $i$ often has the geometrical interpretation of the pseudoscalar $\vol$, as in equation~\eqref{eqn:bivector-spacetime-split}, where both $i$ and $\vol$ play a role similar to the Hodge dual \cite{dressel2015sta}.
	In these cases, $\vol$ is ``the'' principal root of $-1$.
}

\subsection{The space/time split}

While we actually live in $\RR^{1,3}$ spacetime, to any particular observer it appears that space is $\RR^3$ with a separate scalar time parameter.
This is reflected in the fact that $\GA[+](1,3)$ and $\GA(3)$ are isomorphic\footnote{
	An isomorphism of geometric algebras is a linear map $\varphi{}$ respecting the geometric product $\varphi{}(ab) = \varphi{}(a)\varphi{}(b)$.
	Other operations (such as reversion and grade projection) are not necessarily preserved.
} --- in fact there is a distinct isomorphism for each distinct inertial frame's observed spacetime split.
A \emph{space/time split} enables spacetime multivectors to be represented in a frame-dependent manner as $\GA(3)$ multivectors, and is performed as follows.

Suppose $K$ is an inertial observer, and for simplicity choose the standard basis $\qty{\vg \mu{}}$ so that $\vg 0$ is the instantaneous velocity of the $K$ frame.
There is an associated set of \emph{relative vectors} $\vg i\vg 0 \cong{} \vs i$ which form a vector basis for $\GA(3)$ specific to the $K$ frame.\footnote{
	Explicitly, there is an isomorphism $\GA[+](1,3) \cong{} \GA(3)$ constructed in this way for each timelike vector $\vb{v}{} \in{} \GA[1](1,3)$.
	Each isomorphism `splits' even spacetime multivectors into time and space components as observed in the inertial frame with velocity $\vb{v}{}$, providing an efficient, purely algebraic method for switching between inertial frames \cite{hestenes2003sta}.
	Read ``$A \cong B$'' with the understanding that $A \in \GA[+](1,3)$ and $B \in \GA(3)$ are equal under such an isomorphism.
}
For example, with respect to the $K$ frame, a spacetime bivector $F = F^{\mu{}\nu{}}\vg \mu{}\vg \nu{}$ may be separated into timelike $F^{i0}$ and spacelike $F^{ij}$ components and viewed as a $\qty{1,2}$-multivector in $\GA(3)$.
\begin{align}
	F = F^{i0}\vg i\vg 0 + F^{ij}\vg i\vg j
	\cong{} E^i\vs i + B^i\vol\vs i = \vb{E}{} + \vol \vb{B}{}
	\label{eqn:bivector-spacetime-split}
\end{align}
Note that
\begin{math}
	\vg i\vg j
	= (\vg i\vg 0)(\vg j\vg 0)
	\cong{}
	\vs i\vs j
	= \epsilon{}_{ij}{}^k\vol\vs k
\end{math}
where $\vol = \vs1\vs2\vs3$ also denotes the $\GA(3)$ pseudoscalar.\footnote{
	We temporarily assume $\vg0^2 = 1$ for illustration, but of course either metric signature is suitable.
}
This is precisely the frame-dependent decomposition of a spacetime bivector (or ``2-form'') into two $\RR^3$ vectors familiar from electromagnetic theory.

A proper orthochronous Lorentz transformation $\linmap \Lambda{} \in \SO^+(1,3)$ is represented by a rotor $e^\sigma{} \in \Spin^+(1,3)$, which is in turn generated by a spacetime bivector $\sigma{} \in \GA[2](1,3)$.
The bivector $\sigma{} \in \GA[2](1,3)$ in the form of \eqref{eqn:bivector-spacetime-split} is
\begin{align}
	\sigma{} = \frac12(\xi{}^i \vg i + \vol \theta^i \vg i) \wedge \vg 0
	\cong \frac12(\vb \xi{} + \vol \vb \theta)
	\label{eqn:bivector-generator}
,\end{align}
where $\xi{}^i$ and $\theta^i$ are triplets of values representing the three rapidities and three angles which characterise the Lorentz transformation in the $K$ frame.
The rightmost equality shows a space/time split into a rapidity vector $\vb \xi{} \in \RR^3$ and rotation bivector~$\vol\vb \theta$.

The geometric \x{BCH} formula \eqref{eqn:geometric-BCH} as written only involves the geometric product, and so lifts into $\GA(3)$ identically.
(However, \eqref{eqn:geometric-BCH-products} does not translate directly, since the grade-dependent products are not preserved by the isomorphism.)

\subsection{The invariant bivector decomposition}
\label{sec:invariant-bivector-decomposition}

\toself{Might be best to use $F$ consistently as the symbol for a generic spacetime bivector instead of `$\sigma{}$' for ``generator''.}

Spacetime bivectors $\sigma{} \in \GA[2](1,3)$ may always be normalized, in the sense that there exists some $N_\sigma{} \in \GA[0,4](1,3)$ such that
\begin{align}
	\sigma{} = N_\sigma{}\hat{\sigma{}} = \hat{\sigma{}}N_\sigma{}
	\qqtext{where}
	\hat{\sigma{}}^2 = 1
,\end{align}
except in the case $\sigma{}^2 = 0$, where we let $\hat{\sigma{}}^2 = 0$ instead.
This is because the square of a spacetime bivector $\sigma{}^2 = \alpha{} + \vol \beta{} = \rho{}^2e^{2\vol \phi{}}$ always possesses a $\qty{0, 4}$-multivector principal root
\begin{align}
	N_\sigma{}
	\coloneqq \sqrt{\sigma{}^2}
	= \rho{}e^{\vol \phi{}}
	\label{eqn:spacetime-bivector-normalizer}
,\end{align}
assuming without loss of generality that $\rho{} > 0$ and $\phi{} \in (-\pi{}/2, \pi{}/2]$.
The \emph{invariant bivector decomposition}
\begin{align}
	\sigma{} = \rho{}e^{\vol \phi{}}\hat{\sigma{}}
	&= \underbrace{(\rho{}\cos \phi{})\hat{\sigma{}}}_{\sigma{}_+} + \underbrace{(\rho{}\sin \phi{})\vol\hat{\sigma{}}}_{\sigma{}_-}
\end{align}
also defined in \cite[\textsection{}\,5.4.1]{doran2003ga} and \cite{hestenes2003sta} separates $\sigma{}$ into commuting parts $[\sigma{}_+, \sigma{}_-] = 0$ each of which satisfy $\pm{}\sigma{}_\pm{}^2 > 0$.

This decomposition makes clear the non-injectivity of the exponential map.
For instance, each bivector in the family
\begin{math}
	\sigma{}_{n} = \lambda{}_+\hat \sigma{} + (\lambda{}_- + n\pi{})\vol\hat \sigma{}
	\label{eqn:equivalent-generators}
\end{math}
generates the same Lorentz rotor up to an overall sign,
\begin{align}
	e^{\sigma{}_{n}} = e^{\sigma{}_{0}}e^{n\pi{}\vol\hat \sigma{}} = (-1)^ne^{\sigma{}_{0}}
	\label{eqn:equivalent-rotors}
,\end{align}
and all such rotors correspond to the same Lorentz transformation of vectors.
The equivalence~\ref{eqn:equivalent-rotors} shows that every Lorentz rotor $\pm{}e^{\sigma{}_0}$ is equal to a pure bivector exponential $e^{\sigma{}_n}$ with a shifted rotational part $\lambda{}_- \mapsto \lambda{}_- + n\pi$.
The \x{BCH} formula \eqref{eqn:geometric-BCH} discards overall sign, so assuming $|\lambda{}_-| \leq{} \pi{}$ then $\ln(e^{\sigma{}_{n}}) = \sigma{}_{0}$ using the standard branch cut $\qty|\IM{\tanh^{-1} z}| \leq{} \frac{\pi{}}{2}$.

\subsection{The BCHD formula in Minkowski spacetime}

Because the geometric \x{BCH} formula is constructed from sums and products of bivectors, it involves only even spacetime multivectors.
Therefore, in numerical applications, it is not necessary to represent the full STA, but only the even sub-algebra $\GA[+](1,3) \cong \GA(3)$.
The algebra of physical space $\GA(3)$ admits a faithful complex linear representation by the Pauli spin matrices \cite{baylis1989sta-pauli,lasenby2016ga-unified-language,hestenes2003sta}.
The real dimension of both $\CC^{2\times{}2}$ and $\GA(3)$ is eight, so there is no redundancy in the Pauli representation, so it is convenient for computer implementation.

An even $\GA[+](1,3)$ multivector --- or equivalently, a general $\GA(3)$ multivector --- may be parametrised by four complex scalars $q^\mu{} = \RE{q^\mu{}} + i\IM{q^\mu{}} \in \CC$ as
\begin{align}
	A = \RE{q^0} + \RE{q^i}\vs i + \IM{q^i}\vol\vs i + \IM{q^0}\vol
,\end{align}
where the $\vs i$ may be read both as spacetime bivectors $\vs i \equiv \vg 0\vg i \in \GA[+](1,3)$ or as basis vectors of $\GA(3)$ under a space/time split.
The Pauli matrices $\sigma{}_i \in{} \CC^{2\times{}2}$ form a linear representation of $\GA(3)$ by the association $\vs i \equiv \sigma{}_i$.
Explicitly, identifying
\begin{align}
	\vs1 &\equiv \mqty[
		 0&+1\\
		+1& 0\\
	]
&	\vs2 &\equiv \mqty[
		 0&-i\\
		+i& 0\\
	]
&	\vs3 &\equiv \mqty[
		+1& 0\\
		 0&-1\\
	]
\end{align}
along with $1 \equiv I$ and $\vol \equiv iI$ where $I$ is the identity matrix, we obtain a representation of the multivector $A$ by a $2 \times{} 2$ Hermitian matrix:
\begin{align}
	\linmap A \equiv \mqty[
		q^0 + q^3 & q^1 - iq^2 \\
		q^1 + iq^2 & q^0 - q^3
	]
	\label{eqn:multivector-as-matrix}
.\end{align}

A proper Lorentz transformation $\linmap \Lambda{} \in \SO^+(1,3)$ is determined in the $K$ frame by a vector rapidity $\vb \xi{} \in \RR^3$ and axis-angle vector $\vb \theta \in \RR^3$.
The standard $4\times{}4$ matrix representation of $\linmap \Lambda{}$ is obtained as the exponential of the generator
\newcommand{\zer}{\phantom\pm0\phantom{^i}} 
\begin{align}
	\mqty[
		 0 &  \vb \xi{}^\intercal \\
		\vb \xi{} & \epsilon{}_{ijk}\theta^k
	] =
	\left[\begin{array}{c|ccc}
	\zer & \xi{}^1 & \xi{}^2 & \xi{}^3 \\
	\hline
	\xi{}^1 & \zer & +\theta^3 & -\theta^2 \\
	\xi{}^2 & -\theta^3 & \zer & +\theta^1 \\
	\xi{}^3 & +\theta^2 & -\theta^1 & \zer
	\end{array}\right]
	\in \so(1,3)
	\label{eqn:lorentz-generator-matrix}
.\end{align}
In the spin representation, the transformation $\linmap \Lambda{}$ corresponds to a rotor $\rotor L = e^\sigma{}$, and the generating bivector \eqref{eqn:bivector-generator} may be expressed via \eqref{eqn:multivector-as-matrix} as the traceless complex matrix
\begin{align}
	\linmap \Sigma{} = q^k\sigma{}_k = \mqty[
		+q^3 & q^1 - iq^2 \\
		q^1 + iq^2 & -q^3
	]
	\label{eqn:bivector-as-matrix}
,\end{align}
where $q^k \coloneqq \frac12(\xi{}^k + i\theta^k) \in \CC$.
Note that, since the square of a spacetime bivector is a $\qty{0,4}$-multivector, its representative matrix $\linmap \Sigma{}$ squares to a complex scalar multiple of the identity.

Given two generators $\sigma{}_i$ with matrix representations $\linmap \Sigma{}_i$, the geometric \x{BCH} formula \eqref{eqn:geometric-BCH} reads in terms of matrix operations,
\begin{align}
	\linmap \Sigma{}_3 :={} \bch{\linmap \Sigma{}_1}{\linmap \Sigma{}_2} = \tanh^{-1}\qty(
		\linmap{\frac{ T_1 + T_2 + A }{ I + S }}
	)
	\label{eqn:matrix-BCH}
,\end{align}
where $\linmap T_i :={} \tanh \linmap \Sigma{}_i$.
To efficiently compute $\linmap T_i$, make use of the fact that $\linmap \Sigma{}_i^2 = \lambda{}_i^2\linmap I$ where $\lambda{}_i \in \CC$ and evaluate $\linmap T_i = (\tanh \lambda{}_i)\lambda{}_i^{-1}\linmap \Sigma{}_i$.
In the null case $\linmap \Sigma{}_i^2 = \lambda{} = 0$, we have trivially $\tanh \linmap \Sigma{}_i = \linmap \Sigma{}_i = \tanh^{-1} \linmap \Sigma{}_i$.

The commutator $\linmap A :={} \frac12[\linmap T_1, \linmap T_2]$ and anti-commutator $\linmap S :={} \frac12\qty{\linmap T_1, \linmap T_2}$ terms may be efficiently computed by separating the single matrix product $\linmap{ \Pi{} :={} T_1T_2 = A + S }$ into off-diagonal and diagonal components, respectively; \emph{i.e.},
\begin{align}
	\linmap A_{ij} = (1 - \delta{}_{ij})\linmap \Pi{}_{ij}
	\qqtext{ and }
	\linmap S_{ij} = \delta{}_{ij}\linmap \Pi{}_{ij}
.\end{align}
The numerator of \eqref{eqn:matrix-BCH} is therefore a matrix with zeros on the diagonal, and the denominator is a complex scalar multiple of the identity, so the argument of $\tanh^{-1}$ (call it $\linmap M$) is of the form \eqref{eqn:bivector-as-matrix}.
Computing $\tanh^{-1} \linmap M$ again simply amounts to
\begin{math}
	\linmap \Sigma{}_3 = \tanh^{-1} \linmap M = (\tanh^{-1} \lambda{})\lambda{}^{-1}\linmap M
\end{math}
where $\linmap M^2 = \lambda{}^2\linmap I$.
The Lorentz generator in the standard vector representation \eqref{eqn:lorentz-generator-matrix} can then be recovered from $\linmap \Sigma{}_3$ with the relations $\xi{}^k = 2\RE{q^k}$ and $\theta^k = 2\IM{q^k}$, and the final $\SO^+(1,3)$ vector transformation is its $4\times{}4$ matrix exponential.

\section{Composition of Relativistic 3-velocities and the Wigner Angle}

As an example of its theoretical utility, we shall use the geometric \x{BCH} formula to derive the composition law for arbitrary relativistic $3$-velocities.
The innocuous problem of composing relativistic velocities has been called ``paradoxical'' \cite{ungar1989sr-velocity-composition,mocanu1992sr-velocity-composition,visser2011sr-velocity-composition}, owing in part to the fact that \emph{irrotational} boosts are not closed under composition, and that it is difficult to make sense of this additional complexity by representing the general composition in explicit matrix form.
Of course, there is no paradox, and the full description of the composition of boosts is pedagogical as it highlights aspects of special relativity which differ from common intuition.

Given an inertial frame $K$, we may speak of \emph{pure rotations} or \emph{pure boosts} relative to the $K$ frame (a pure rotation or pure boost relative to $K$ is \emph{not} pure in all other frames).
The restriction of the \x{BCH} formula to pure boosts is not as simple as the restriction to rotations \eqref{eqn:rodrigues-formula}, because pure boosts do not form a closed subgroup of $\SO^+(1,3)$ like pure rotations do.
Instead, the composition of two pure boosts $\rotor B_i$ is a pure boost composed with a pure rotation (or vice versa),
\begin{align}
	\rotor B_1\rotor B_2 = \rotor B\rotor R
	\label{eqn:boost-boost-makes-boost-rotation}
.\end{align}
The direction of the boost $\rotor B$ lies within the plane defined by the boost directions of $\rotor B_1$ and $\rotor B_2$, and $\rotor R$ is a rotation through this plane by the \emph{Wigner angle} \cite{visser2011sr-velocity-composition}.
Applying \eqref{eqn:geometric-BCH} to this case immediately yields formulae for the resulting boost and rotation.
These results are isomorphic to those in \cite{berry2020quat-sr} which are formulated using complexified quaternions.

For ease of algebra, we conduct the following analysis under a space/time split with respect to the $K$ frame.
Under this split, a pure boost $\rotor B$ is generated by an $\RR^3$ vector $\frac{\vb \xi{}}{2}$, and a pure rotation $\rotor R$ is generated by an $\RR^3$ bivector $\frac\theta2\hat{r}$.
Here, $\vb \xi{} \in{} \GA[1](3)$ is the \emph{vector rapidity}, related to the velocity by $\vb v/c = \vb \beta{} = \tanh \vb \xi{}$, and the rotation is through an angle $\theta$ in the plane spanned by the bivector $\hat{r} \in{} \GA[2](3)$.
Equation \eqref{eqn:geometric-BCH} with two pure boosts $\vb \xi{}_1$ and $\vb \xi{}_2$ is
\begin{align}
	\tanh(\bch{\frac{\vb \xi{}_1}{2}}{\frac{\vb \xi{}_2}{2}})
	= \frac{\vb{w}{}_1 + \vb{w}{}_2 + \vb{w}{}_1 \wedge{} \vb{w}{}_2}{1 + \vb{w}{}_1\cdot{}\vb{w}{}_2}
	\label{eqn:boost-boost}
\end{align}
where $\vb{w}{}_i :={} \tanh\frac{\vb \xi{}_i}{2}$ are the \emph{relativistic half-velocities}, also defined in \cite{berry2020quat-sr,berry2021quat-sr}.
The generator \eqref{eqn:boost-boost} has vector and bivector (namely $\vb{w}{}_1 \wedge{} \vb{w}{}_2$) parts, indicating that the Lorentz transformation it describes is indeed some combination of a boost and a rotation.

Similarly, for an arbitrary pure boost and pure rotation, equation~\eqref{eqn:geometric-BCH} is
\begin{align}
	\tanh(\bch{\frac{\vb \xi{}}2}{\frac\theta2\hat{r}})
	= \frac{\vb{w}{} + \rho{} + \frac12[\vb{w}{}, \rho{}]}{1 + \vb{w}{}\wedge \rho{}}
	\label{eqn:boost-rotation}
\end{align}
where $\rho{} :={} \tanh \frac{\theta\hat{r}}2 = \hat{r}\tan\frac \theta2$ is a bivector.
In general, \eqref{eqn:boost-rotation} has vector, bivector \emph{and} pseudoscalar parts (the commutator $\frac12[\vb{w}{}, \rho{}] = \grade[1]{\vb{w}{}\rho{}} + \vb{w}{} \wedge \rho{}$ and the denominator both have grade-three part $\vb{w}{}\wedge \rho{}$).
However, \eqref{eqn:boost-boost} and \eqref{eqn:boost-rotation} are equal by supposition of \eqref{eqn:boost-boost-makes-boost-rotation}, and by comparing parts of equal grade, we deduce the pseudoscalar part of \eqref{eqn:boost-rotation} is zero.
This enforces $\vb{w}{}\wedge \rho{} = 0$, or equivalently, that $\vb{w}{}$ lies in the plane defined by $\rho{}$ --- meaning the resulting boost lies within the plane of Wigner rotation as expected.
Hence, for a coplanar boost and rotation, \eqref{eqn:boost-rotation} is simply
\begin{align}
	\tanh(\bch{\frac{\vb \xi{}}2}{\frac\theta2\hat{r}})
	= \vb{w}{} + \rho{} + \vb{w}{}\rho{}
	\label{eqn:boost-rotation-coplanar}
.\end{align}
The term $\vb{w}{}\rho{} = \grade[1]{\vb{w}{}\rho{}} = -\rho{}\vb{w}{}$ is a vector orthogonal to $\vb{w}{}$ in the plane defined by $\rho{}$.

Equating the bivector parts of \eqref{eqn:boost-boost} and \eqref{eqn:boost-rotation-coplanar} determines the rotation
\begin{align}
	\rho{} &= \frac{\vb{w}{}_1 \wedge \vb{w}{}_2}{1 + \vb{w}{}_1 \cdot{} \vb{w}{}_2}
,&	&\text{implying}
&	\theta = 2\tan^{-1}\qty(\frac{w_1w_2\sin\phi{}}{1 + w_1w_2\cos\phi{}})
\end{align}
where $\phi{}$ is the angle between the two initial boosts (in the $K$ frame).
The angle $\theta$ is precisely the Wigner angle.
Equating the vector parts determines the boost
\begin{align}
	\vb{w}{} &= \frac{\vb{w}{}_1 + \vb{w}{}_2}{1 + \vb{w}{}_1\cdot{}\vb{w}{}_2}(1 + \rho{})^{-1}
,\end{align}
noting that $\vb{w}{}_i$ and $\rho{}$ do not commute.
Substituting $\rho{}$ leads to the remarkably succinct composition law
\begin{math}
	\vb{w}{} = (\vb{w}{}_1 + \vb{w}{}_2)(1 + \vb{w}{}_1\vb{w}{}_2)^{-1}
\end{math}
exhibited in \cite{berry2020quat-sr}, with the final relativistic velocity being $\vb \beta{} = \tanh \vb \xi{} = \tanh (2\tanh^{-1} \vb{w}{})$.

\section{Conclusions}

In geometric algebras of dimension $p + q \leq{} 4$, orthogonal transformations $a \mapsto e^\sigma{}ae^{-\sigma{}} \in \SO^+(p, q)$ may be composed in terms of their generators using the geometric \x{BCH} formula \eqref{eqn:geometric-BCH}, which satisfies
\begin{align}
	e^{\sigma{}_1}e^{\sigma{}_2} = \pm e^{\bch{\sigma{}_1}{\sigma{}_2}}
	\label{eqn:BCH-application}
.\end{align}
This holds for bivectors $\sigma{}_i \in \GA[2](p, q)$, generalizing Rodrigues' formula --- but also for arbitrary $\qty{1,2}$-multivectors in $\GA(3)$, by exploiting the space/time split $\GA[+](1,3) \cong \GA(3)$.
Representing $\GA(3)$ by $2\times{}2$ complex matrices results a computationally efficient formula \eqref{eqn:matrix-BCH} for the composition of proper Lorentz transformations in terms of their generators.

The benefit of adopting geometric algebra here is the utility of the double-cover spin representation: the geometric \x{BCH} formula is simpler than previous results \cite{coll1990sr-generator-exponentiation,coll2002sr-generator-composition} formulated in terms of 2-forms.
It is sufficiently simple to be of theoretical use: it easily reduces to well-known formulae in lower dimensions, and yields the composition law \cite{berry2020quat-sr} for relativistic $3$-velocities and the associated Wigner angle \cite{berry2021quat-sr,visser2011sr-velocity-composition}.

\appendix
\section{Geometric Algebras in Physics}
\label{apx:geoalg}

For any real vector space with a metric $(V, \eta{})$, there is a unique \emph{geometric algebra} $\GA(V, \eta{})$.
``Geometric algebra'' is a synonym for \emph{Clifford algebra} $Cl(V, q)$ in the case that the vector space $V$ is real and is provided with a quadratic form.
The prescription of a quadratic form\footnote{
	A quadratic form $q : V \to \RR$ satisfies $q(\lambda{}\vb{u}) = \lambda{}^2q(\vb{u})$ and measures the (possibly negative) squared norm of a vector.
	The associated metric satisfying $\ip{\vb{u}}{\vb{u}} = q(\vb{u})$ is uniquely recovered from
	\begin{math}
		\ip{\vb{u}}{\vb{v}} = \frac12(q(\vb{u} + \vb{v}) - q(\vb{u}) - q(\vb{v}))
	.\end{math}
	
} $q$ is equivalent to a choice of metric $\eta{}$, but the notion of a metric is more common in physics (whereas the mathematical viewpoint often starts with $q$).
While for its pure mathematical study Clifford's name is retained, the name \emph{geometric algebra} emphasising its rich geometric interpretation is preferred in application to physics.

Succinctly put, $\GA(V, \eta{})$ is obtained by allowing vectors in $V$ to be multiplied freely to form objects of higher \emph{grade}, modulo the identification
\begin{align}
	\label{eqn:vector-square}
	\vb{u}^2 = \ip{\vb{u}}{\vb{u}} = \eta{}_{ab}u^au^b \in \RR
\end{align}
of the square of any vector with its scalar inner product.
This rule completely defines the \emph{geometric product} which we denote by juxtaposition.
The resulting $2^n$-dimensional algebra is \emph{graded}: as a vector space is isomorphic to the exterior algebra $\bigwedge(V)$ with a $\binom{n}{k}$-dimensional subspace for each grade $k$.
However, $\GA(V, \eta{})$ is a metric-dependent generalisation of the exterior algebra (and is not usually defined on the dual space $V^*$ as $p$-forms are).
Also, unlike $\bigwedge(V)$, objects of mixed grade in  $\GA(V, \eta{})$ play the extremely useful role of describing reflections and rotations in arbitrary dimensions.

By expanding $(\vb{u}{} + \vb{v}{})^2 = \langle\vb{u}{} + \vb{v}{}, \vb{u}{} + \vb{v}{}\rangle$, we immediately find
\begin{align}
	\langle\vb{u}{}, \vb{v}{}\rangle = \frac12(\vb{u}{}\vb{v}{} + \vb{v}{}\vb{u}{})
.\end{align}
So the symmetric part of a product of vectors is their inner product; a scalar, or grade zero quantity.
The antisymmetric part coincides with the alternating wedge product familiar to exterior algebra (only now defined on vectors, not co-vectors)
\begin{align}
	\vb{u}{} \wedge{} \vb{v}{} = \frac12(\vb{u}{}\vb{v}{} - \vb{v}{}\vb{u}{})
.\end{align}
This is a grade $2$ object, or bivector, dual to a $2$-form.
Therefore, for the geometric product of vectors we have the famous relation
\begin{align}
	\vb{u}{}\vb{v}{} = \langle\vb{u}{}, \vb{v}{}\rangle + \vb{u}{}\wedge{}\vb{v}{}
\end{align}
and it follows that parallel vectors commute and orthogonal vectors anticommute.

We denote by $\GA(p, q)$ the geometric algebra over $V = \RR^{p,q}$, which then admits an orthonormal basis
\begin{math}
	\qty{\vb{e}{}^+_1, ..., \vb{e}{}^+_p, \vb{e}{}^-_1, ..., \vb{e}{}^-_q}
	\text{ with }
	\qty(\vb{e}{}^\pm{}_i)^2 = \pm{}1
.\end{math}
A basis of the entire algebra consists of
\begin{center}
\begin{tabular}{rl}
	vectors or $1$-blades & $\vb{e}{}_i, \quad i \in{} \qty{1, 2, \dots, n}$
\\	$2$-blades & $\vb{e}{}_i\vb{e}{}_j = -\vb{e}{}_j\vb{e}{}_i, \quad i \not ={} j$
\\	$\vdots$
\\	$k$-blades & $\vb{e}{}_{i_1}\vb{e}{}_{i_2}\cdots\vb{e}{}_{i_k} = \epsilon{}^{j_1j_2\cdots j_k}\vb{e}{}_{i_{j_1}}\vb{e}{}_{i_{j_2}}\cdots\vb{e}{}_{i_{j_k}}$
\\	$\vdots$
\end{tabular}
\end{center}
and so on up to the \emph{pseudoscalar} $\vol :={} \vb{e}{}_1\vb{e}{}_2\cdots \vb{e}{}_n$.

A \emph{$k$-vector} is a sum of $k$-blades, and a $k$-blade is a $k$-vector which is expressible as the wedge product of $k$ vectors.\footnote{The simplest example of a $2$-vector which is not a $2$-blade is $\vb{e}{}_1\vb{e}{}_2 + \vb{e}{}_3\vb{e}{}_4 \not ={} \vb{u}{}\wedge{}\vb{v}{}$.}
A general element of the algebra (of uniform or mixed grade) is called a (\emph{homogeneous} or \emph{inhomogeneous}) \emph{multivector}.
Finally, if the non-zero parts of a multivector $A$ have grade $k \in K \subseteq \NN$ for some set of grades $K$, we shall call $A$ a $K$-multivector.

\subsection{Fundamental dualities of a geometric algebra}

Linear operations such as the matrix transpose or complex and hermitian conjugates are useful because they preserve (or reverse) multiplication: they are (anti-)auto\-morphisms.
Geometric algebras possess two distinguished automorphisms:
\begin{itemize}
	\item \emph{Grade involution, $\iota{}$.}
	Reflection through the origin $\iota{}(\vb{u}) = -\vb{u}$ is an isometry and so extends to an algebra automorphism by the requirement $\iota{}(ab) = \iota{}(a)\iota{}(b)$.
	Its action on $k$-vectors is
	\begin{math}
		\iota{}(a) = (-1)^ka
	\end{math}
	and is defined on multivectors by linearity.

	\item \emph{Reversion, $\rev{\phantom{a}}$.}
	The \emph{reverse} flips the order of the geometric product,
	\begin{math}
		\rev{ab} = \rev{b}\rev{a}
	\end{math}
	(making it an anti-automorphism) and is the identity on vectors, $\rev{\vb{u}{}} = \vb{u}{}$.
	Explicitly, if $a$ is a $k$-blade, then $\rev{a} = s_ka$, where $s_k = \epsilon{}^{k\cdots 21} = (-1)^{\frac{(k - 1)k}2} = \pm1$ is the sign of the reverse permutation on $k$ symbols.
\end{itemize}
Grade involution fixes even-grade elements, which together form the \emph{even sub-algebra}
\begin{align}
	\GA[+](p, q) \coloneqq \qty{a \in \GA(p, q) \mid \iota{}(a) = a}
.\end{align}
The even sub-algebra is algebraically closed, and can generally be interpreted as the \emph{rotation algebra} for $\RR^{p,q}$.
Elements which are their own reverse $\rev{a} = a$ are sums of blades of grade $k \in \qty{4n, 4n + 1 \mid n \in \NN} = \qty{0, 1, 4, 5, 8, 9, \dots}$ only.

These operations are useful in practice.
In particular, the following result follows easily from reasoning about grades.
\begin{lemma}
	\label{lem:grades-of-square}
	If $A$ is a $k$-vector, then $A^2$ is a $4\NN$-multivector, \emph{i.e.}, a sum of blades of grade $\qty{0, 4, 8, \dots}$ only.
\end{lemma}
\begin{proof}
	The multivector $a^2$ is its own reverse, since $\rev{a^2} = (\rev{a})^2 = (\pm a)^2 = a^2$, and hence has parts of grade $\qty{4n, 4n + 1 \mid n \in \NN}$.
	Similarly, $a^2$ is self-involutive, since $\iota{}(a^2) = \iota{}(a)^2 = (\pm a)^2 = a^2$, and is thus of even grade.
	Therefore $a^2$ is a $\qty{0, 4, 8, ...}$-multivector.
\end{proof}

\subsection{Relationships to other common algebras}

Geometric algebra reproduces many of the useful algebraic structures found in physics.
Complex numbers are fit for describing $\SO(2)$ rotations; quaternions for $\SO(3)$ rotations in $\RR^3$; and work has been done with complexified quaternions $\CC\otimes\HH$ in describing Lorentz transformations \cite{berry2020quat-sr,berry2021quat-sr}.
All these algebras are isomorphic to an even geometric sub-algebra
\begin{align}
	\CC &\cong \GA[+](2)
,&	\HH &\cong \GA[+](3)
,&	\CC\otimes\HH &\cong \GA[+](1,3)
,\end{align}
where the role of conjugation is played by reversion.
Common to all these isomorphisms is the identification of each ``imaginary'' unit with a unit bivector $\ve{i}\ve{j}$.
In 2d, there is one linearly independent bivector, $\ve1\ve2$, and one imaginary unit, $i$.
Indeed, in 3d, there are $\binom{3}{2} = 3$ bivectors, and so three imaginary units $\qty{\vb i, \vb j, \vb k}$ are needed.

The interpretation of a bivector is clear: it generates a rotation in the oriented plane which it spans.
That imaginary units are best interpreted as bivectors (or ``axial'' vectors), and not as ordinary (``polar'') vectors, reveals some of the confusion that surrounds the quaternions \cite{chappell2016quat-history,altmann1989quat-history}.
It is only a happy (or misleading) coincidence that in $\RR^3$ vectors and bivectors can be interchanged---but not without sacrificing proper transformation behaviour.
This is also why the complex numbers do not represent vectors in an isotropic\footnote{In a way that treats all directions on equal footing.} way: $\CC$ does not contain ordinary vectors; it is the linear combination of one scalar $1$ and one bivector $i$.

Enlarging $\GA[+](p, q)$ to the full algebra $\GA(p, q)$ adds the missing polar 1-vectors along with other objects of odd grade.
Such an algebra describes vectors and rotations in a unified and isotropic way.
Physics has independently invented $\GA(p, q)$ in at least two instances in the form of the Pauli and Dirac matrix algebras.
The Pauli matrices $\qty{\sigma{}_1, \sigma{}_2, \sigma{}_3}$, satisfying $\qty{\sigma{}_i, \sigma{}_j} = 2\delta{}_{ij}I$ form a faithful complex linear representation of $\GA(3)$.
Likewise, their relativistic counterpart the Dirac matrices $\qty{\gamma{}_0, \gamma{}_1, \gamma{}_2, \gamma{}_3}$ are algebraically isomorphic to $\GA(1,3)$.
In both cases, hermitian conjugation in the matrix algebra corresponds to reversion in the geometric algebra \cite[\textsection{}\,5]{doran2003ga}.

\acknowledgments

Joseph Wilson was supported by a Victoria University of Wellington MSc scholarship, and was also indirectly supported by the Marsden Fund, via a grant administered by the Royal Society of New Zealand.
Matt Visser was directly supported by the Marsden Fund, via a grant
administered by the Royal Society of New Zealand.

Joseph would like to thank Peter Donelan of Victoria University of Wellington for his helpful input.

\addcontentsline{toc}{section}{Bibliography}

\bibliography{references}

\end{document}